\documentclass[11pt]{article}
\def\submission{0} 

\usepackage[utf8]{inputenc}
\usepackage{setspace}
\usepackage{fullpage}
\usepackage{graphicx}
\usepackage{bbm}
\graphicspath{ {./figures/} }
\usepackage{subcaption}
\usepackage{amsmath,amsfonts,amssymb,amsthm}
\usepackage{xcolor}
\usepackage[colorlinks=true,
            citecolor=blue]{hyperref}
\usepackage{fullpage}
\usepackage{algorithm}
\usepackage{algpseudocode}
\usepackage{bbm}
\usepackage{pgfplots}
\pgfplotsset{compat=1.16}
\usepackage{mathtools}
\usepackage{float}
\usepackage{indentfirst}
\usepackage{tikz}
\usepackage[capitalize]{cleveref}
\usepackage{accents}

\usepackage{framed}

\usetikzlibrary{external, matrix, positioning, patterns, arrows.meta, decorations.markings, decorations.pathmorphing, calc, shapes.misc}
\tikzset{cross/.style={cross out, draw=black, minimum size=2*(#1-\pgflinewidth), inner sep=0pt, outer sep=0pt},
cross/.default={1pt}}
\allowdisplaybreaks
\newtheorem{theorem}{Theorem}[section]

\newtheorem{lemma}{Lemma}[section]

\newtheorem*{prop*}{Proposition}
\newtheorem*{lemma*}{Lemma}
\newtheorem{conjecture}{Conjecture}[section]
\newtheorem*{definition*}{Definition}
\newtheorem*{theorem*}{Theorem}

\theoremstyle{definition}
\newtheorem{definition}{Definition}[section]

\theoremstyle{remark}
\newtheorem{remark}{Remark}[section]
\numberwithin{equation}{section}
\usepackage{xspace}
\crefname{prop}{Proposition}{Propositions}

\makeatletter
\AddToHook{cmd/appendix/before}{\def\cref@section@alias{appendix}\def\cref@subsection@alias{appendix}}
\makeatother

\newcommand{\qsample}{\operatorname{QSAMPLE}}
\newcommand{\sample}{\operatorname{SAMPLE}}
\newcommand{\pex}{\operatorname{PEX}}
\newcommand{\etal}{\emph{et al.\@}}

\newcommand{\Bin}{\mathrm{Bin}}
\newcommand{\poly}{\mathrm{poly}}
\newcommand{\TV}{\mathrm{TV}}

\usepackage[backend=biber,style=alphabetic,maxbibnames=99,maxalphanames=99,url=false]{biblatex}
\providecommand{\email}[1]{\href{mailto:#1}{\nolinkurl{#1}\xspace}}

\title{A Quantum/Classical Example Oracle Separation for Making Things Up}

\ifnum\submission=0 
\author{Kenny Chen\thanks{The University of Sydney, Email: \email{kche5493@uni.sydney.edu.au}}}
\else 
\author{Author(s) redacted}
\fi

\date{August 2026}

\begin{document}

\maketitle

\begin{abstract}
Consider two PAC learning algorithms, both having access to quantum computation, but differing in the types of examples they obtain: one is provided with classical samples, while the other is given quantum samples. Are there any learning tasks that can be efficiently performed by the latter, but not by the former? This question, the focus of our work, is surprisingly still open. Our main result is to show that \emph{relative to an oracle}, there are distributions that can be efficiently generated by a quantum learner with access to quantum samples, but not by a quantum learner with access to only classical samples, making progress to answering this question in the affirmative.
\end{abstract}
\section{Introduction}
In 1984, Valiant introduced the celebrated concept of Probably Approximately Correct (PAC) Learnability \cite{DBLP:journals/cacm/Valiant84}, which then sparked the fields of computational learning theory, machine learning theory, and influenced many others. A major question in this field then, was to determine what functions, or classes of functions were efficiently PAC learnable. A decade later, Shor presented his seemingly unrelated yet equally celebrated quantum algorithm for factoring integers \cite{DBLP:journals/siamcomp/Shor97}. This spurred yet another field of questions, centered around the power of quantum computation. Namely, how much more powerful were quantum computers than classical computers, and what other problems admitted quantum speedups. Soon, these two fields intertwined, leading to the question of whether there was quantum advantage in PAC learning.

In the PAC model, a learner receives examples $(x,f(x))$ from an unknown function $f$ drawn from a concept class, and must output a hypothesis function approximating $f$. We will also consider the closely related PAC generative (PAC GEN) model, where instead of functions, we are concerned with distributions. Here, an unknown distribution $\mathcal{D}$ is chosen from a distribution class, and a learner receives samples from it, and is tasked to output an efficient generator whose output distribution is close to $\mathcal{D}$ (see \cref{def:pac_learning_functions} and \cref{def:pac_gen_learning} for details). We can just the efficiency of learners in two ways. First, we can require that they be sample efficient, meaning that they have polynomial sample complexity. Secondly, which is the one this paper moreso concerns itself with is that the learner is polynomial both in terms of runtime and sample complexity, in which case it will be deemed ``efficient''. Formalizing the question of quantum advantage, one could ask the following motivating question:
\begin{center}
    \textit{Are there function/distribution classes that are efficiently by a quantum learner, but not by a classical learner?}
\end{center}

It is worth examining the two quantum resources that can potentially lead to quantum advantage. The first is from quantum \emph{computation}, where the advantage comes from the quantum computer using qubits, and the associated quantum gates that are hard to simulate for a classical computer. The second is from quantum samples, where a quantum learner gets quantum access to the unknown distribution or functions, whereas the classical learner only has classical access. There is, in general, no unique way of defining what quantum access means. We refer to Gily\'{e}n and Li's introduction for a survey of these access models \cite{DBLP:conf/innovations/GilyenL20}, but we use the \emph{sampling} model of access. At a high level, this means the quantum learner gets access to copies of a state $\rho$ which encodes a superposition over the function or distribution (see \cref{def:oracles} for details). This is one of the weakest forms of access defined by Gily\'{e}n and Li, and also the form considered by previous works in quantum PAC learning, such as by \cite{DBLP:journals/qip/AticiS07} and \cite{DBLP:journals/quantum/SwekeSHE21}.

This separation between quantum and classical learners has been studied previously. Starting with the separations in sample complexity, Arunachalam and De Wolf showed that in the most general formulation of PAC learning under arbitrary distirbutions, quantum learners could only admit at most a constant factor advantage in sample complexity \cite{Arunachalam2018}. Salmon \etal~later revisited this result, and showed if a quantum learner had access to not just the quantum samples, but also the underlying circuit that generated these samples, then one could get a quadratic advantage in sample complexity \cite{salmon24provable}. This type of access is exactly one considered by Gily\'{e}n and Li, who show that there are property testing problems which only admit quantum advantage if one has access to the quantum circuit, as opposed to simple samples \cite{DBLP:conf/innovations/GilyenL20}. Considering a more restricted setting, Bshouty and Jackson showed that DNF formulae were efficiently PAC learnable by a quantum learner. No efficient classical learner is known (though there is no impossibility result saying one cannot exist) \cite{BShouty1995}. However, these results are not directly applicable to the question in this paper for a few reasons. First, they do not directly make the distinction between quantum computation and quantum samples. They compare a learner with quantum computation and samples samples against a learner with classical computation and samples, and show that generally, the learner with quantum resources has some sort of advantage over the classical counterpart. However, the source of this advantage is not isolated. Secondly, the types of advantages proven are at most quadratic. This means if there is an efficient quantum learner for the above problems, there also exists an efficient (albeit less so) classical learner. Hence, this concerns advantages for tasks that were already efficient for classical learners, rather than demonstrating a separation between quantum and classical learners.

To the best of our knowledge, the first separation result came from Sweke \etal, who showed that quantum \emph{computation} was stricly more powerful than classical computation, specifically in the PAC GEN model \cite{DBLP:journals/quantum/SwekeSHE21}. Specifically, the exhibited a distribution class which, with classical examples, could be efficiently generated by a learner with quantum computation, but not by a classical learner. They left open the question of whether quantum \emph{samples} were strictly more powerful than classical examples in this setting (see \cref{fig:classical_quantum_questions}. That is, given two quantum learners, where one has quantum samples, but another with only classical examples, did there exist functions or distributions that could be efficiently learned or generated by the former, but not by the latter. Though they did not directly address this question, Arunachalam and De Wolf's results which showed that a quantum learner in general, in arbitrary PAC learning settings, admitted only a constant factor advantage, highly suggests that for learning functions, this answer to this question is no. As such, we turn away from PAC learning of functions, and instead to distribution learning to explore possible avenues for advantage. In particular, we resolve to answer the following question, the main topic of this paper:
\begin{framed}\itshape
    \noindent{}How much more powerful are quantum \emph{samples}, compared to classical samples, for distribution learning?
\end{framed} 
Though their focus was on proving computation separations, Sweke \etal~did make several comments and directions for how a separation between quantum and classical examples might be proven. We discuss this in detail in \cref{sec:sweke_context}, but provide a high-level overview here. One key result they showed was the following:
\begin{theorem}[\cite{DBLP:journals/quantum/SwekeSHE21} Informal, see \cref{thm:function_learnability_implies_distribution_learnability}]
    If a function class is efficiently PAC learnable, then the induced distribution class of tuples $(x,f(x))$ can be efficiently generated.
\end{theorem}
To briefly explain the above, PAC learning requires an algorithm, after obtaining samples from an unknown function, to output a function that ``approximates'' the unknown function (for some appropriate notion of approximating). On the other hand, in the generating setting, one only requires that, given samples from an unknown distribution, an algorithm can efficiently output samples from a distribution that is ``close'' to the unknown distribution. Whilst one way of doing this in the above setting, where the underlying distribution is induced by a function, is to learn a function, this is not the only available method in the generating setting. They left the reverse direction open as a conjecture:
\begin{conjecture}[\cite{DBLP:journals/quantum/SwekeSHE21}, Informal, see \cref{conj:function_hardness_implies_distribution_hardness}]
    \label{conj:function_hardness_implies_distribution_hardness:intro}
    If a function class is not efficiently PAC learnable, then the induced distribution class cannot be efficiently generated.
\end{conjecture}
 If proven true, this would potentially allow one to lift hardness results about a function class directly into a quantum/classical computation separation, provided the function class itself was efficiently quantumly learnable, but not efficiently classical learnable (both with classical examples). Of course, the above results are purely classical, but it could potentially provide insights that allowed the result to be lifted into a quantum separation. As such, resolving this conjecture positively could be thought of as a first step that could be leveraged to answer our question.
\begin{figure}[h]
\centering
\begin{tikzpicture}[
    label/.style={font=\small, align=center}
]

\draw (0,0) rectangle (5,5);
\draw (2.5,0) -- (2.5,5);
\draw (0,2.5) -- (5,2.5);

\fill[gray!25] (0,0) rectangle (2.5,2.5);
\draw[pattern=north east lines] (0,0) rectangle (2.5,2.5);

\node[label] at (1.25,5.9) {Classical\\Computation};
\node[label] at (3.75,5.9) {Quantum\\Computation};

\node[label] at (-1.5,3.75) {Classical\\Samples};
\node[label] at (-1.5,1.25) {Quantum\\Samples};

\node at (2.5,3.75) {\LARGE $<$};
\node[above=4pt] at (2.35,3.75) {\large Sweke \etal};

\node at (3.75,2.5) {\LARGE $\wedge$};
\node[right=3pt] at (3.75,2.5) {\Large ?};

\end{tikzpicture}
\caption{Sweke \etal~showed a separation between learners with classical computation and learners with quantum computation when both had only classical examples. It is an open question as to whether there is a separation between quantum learners when one only has classical examples and the other has quantum examples.}
\label{fig:classical_quantum_questions}
\end{figure}
\subsection{Our Results}
Our first result makes progress towards resolving \cref{conj:function_hardness_implies_distribution_hardness:intro} in the negative. Namely, we show that under standard and widely accepted cryptographic assumptions, there is an oracle relative to which  some function class that cannot be efficiently learned, yet the resulting distribution class can be efficiently generated.
\begin{theorem}[Informal, see \cref{thm:function_hardness_does_not_imply_distribution_hardness}]\label{inf:function_hardness_does_not_imply_distribution_hardness}
    Relative to an oracle encoding a one-way permutation $f$, there exists a concept class of functions $\mathcal{C}_f$, which cannot be efficiently learned. However, for any $g\in \mathcal{C}_f$, the distribution class induced by these functions $(x,g(x))$, can be efficiently generated.
\end{theorem}
This strongly suggests that new ideas are required in order to establish a separation for quantum and classical examples. Our second main contribution aims to resolve this question. Namely, we show that given access to an auxiliary function encoded via a classical oracle, there exists a distribution class that can be efficiently PAC generated by a quantum learner with quantum examples, but not by a quantum learner with only classical examples.
\begin{theorem}[Informal, see \cref{thm:qsample_gennable_sample_not_learnable}]\label{inf:qsample_gennable_sample_not_learnable}
    There exists an oracle encoding $g\colon\{0,1\}^n\to\{0,1\}^n$, a random function, and an associated distribution class $\mathcal{C}_g $ 
    that can be efficiently generated by a quantum learner with quantum examples, when given oracle access to $g$. However, no quantum learner with only classical examples (and with oracle access to $g$) can efficiently generate this distribution class.
\end{theorem}
Thus, we make significant progress on resolving the question of whether quantum samples are strictly more powerful than classical samples. In particular, whilst previous separations are with regards to the best known classical algorithms, or proven with respect to well studied assumptions, to the best of our knowledge, this is the first demonstration of unconditional advantage.

\paragraph{A Comment on PAC Learning Outputs.} We comment on what it means for a learner to ``output a hypothesis function efficiently'', as required in \cref{inf:function_hardness_does_not_imply_distribution_hardness}. There are two main ways in which a learner might satisfy this requirement. First, one might require that the PAC learner only efficiently outputs a complete description of a hypothesis, without needing them to actually be able to evaluate the hypothesis themselves. On the other hand, there is also the more stringent interpretation that the PAC learner outputs a circuit that efficiently evaluates the hypothesis function. In our case, we focus on the second one. Though this is not a commonly made distinction for classical PAC learning, it is very relevant in previous contexts of quantum \emph{computation} separations, for example, by Molteni \etal~\cite{ICLR2026_f0075fe4}. This is because one can show quantum computation advantages in learning by generally leveraging the inability of classical learners to evaluate some quantum function, even though they can concisely describe it. In some sense then, the separations which come from the second model of learning are less ``interesting'', and thus work has tended recently to focus on the former model. We discuss this more in \cref{ss:related}. Although this does not directly impact our results, since we deal with purely classical learners for this section, it is worth emphasizing that this is not the source of difficulty for our case regardless. In our case, the hardness separation really comes from the observation that in learning the concept class, the learner must indirectly invert the one-way permutation. Though one can make the same objection regarding our counterexample, where the learner can actually efficiently output a description of the circuit that evaluates the hypothesis, though the circuit cannot be ran efficiently, we emphasize that in this case, the distinction is that the corresponding learner which attempts to generate the distribution still cannot invert the one way permutation, but it simply does not need to. In short, for our specific example, any learner that outputs a good hypothesis must indirectly invert the one way permutation, whereas any learner that generates the distribution does not need to. This is distinct to the above quantum/classical separations, where advantage comes from quantum learners being able to efficiently evaluate some quantum function, whereas classical learners could not. Plausibly, our results could be lifted into the former model, but since our main result concerns a separation between distribution learning, we leave this direction open.

\paragraph{On the access to an auxiliary oracle.} 
Both our main results rely on auxiliary oracle access to another function $g$. In other words, we have an unknown but queryable function $g$, which then induces a concept class $\mathcal{C}_g$. Then the learning algorithms can utilize both examples from the concept classes as well as queries to the function $g$ in order to PAC learn/generate their concept classes. The resulting generators can also use oracle access to $g$ to generate their outputs. We illustrate this interaction in \cref{fig:pac_setting}. While endowing the learner and hypothesis with such an oracle access is not a standard assumption in PAC learning, we note that this has become a standard assumption in the field of quantum query complexity (and also in general complexity theory): indeed, proving oracle-less separations is one of the questions at the forefront of quantum complexity theory, and we currently lack the techniques to generally do so. As such, one standard assumption is to try and prove separations relative to an oracle. In fact, many separations between quantum and classical algorithms are oracle separations of this type, such as most boolean query separations, separations between \textbf{QMA} and \textbf{QCMA} \cite{4262757, fefferman_et_al:LIPIcs.MFCS.2018.22, Natarajan2024distributiontesting} and the oracle separation between \textbf{BQP} and \textbf{PH} \cite{DBLP:journals/jacm/RazT22}, to name a few. Hence, since our goal is to effectively show a separation between quantum and classical samples for the purpose of learning, we believe it is natural to show this separation in the context of an oracle separation.

\begin{figure}[htbp]
\centering
\begin{tikzpicture}[
    node distance=2cm and 2.5cm,
    box/.style={draw, rectangle, minimum width=2cm, minimum height=0.8cm, font=\small},
    bigbox/.style={draw, rectangle, minimum width=3cm, minimum height=1.5cm, font=\small},
    arrow/.style={->, thick},
    edgelabel/.style={midway, fill=white, draw=none, inner sep=1pt, align=center, font=\footnotesize},
]

\node[box] (g) {Auxiliary function $g$};

\node[box] (Qg) [above right=of g] {Oracle $O_g$};
\node[box] (Cg) [below right=of g] {Concept class $\mathcal{C}_g$};

\node[box] (alg) [right=3.5cm of g] {PAC Algorithm};

\node[box] (out) [right=of alg] {Output Function/Generator};

\draw[arrow] (g) -- node[edgelabel, above] {$g$ encoded via $O_g$} (Qg);
\draw[arrow] (g) -- node[edgelabel, below] {$g$ induces $\mathcal{C}_g$} (Cg);

\draw[arrow] (Qg) -- node[edgelabel, above] {oracle access to $g$} (alg);
\draw[arrow] (Cg) -- node[edgelabel, below] {examples from $\mathcal{C}_g$} (alg);

\draw[arrow] (alg) -- node[edgelabel, above] {output} (out);

\draw[arrow] (Qg) to[bend left=20] node[edgelabel, above] {still has oracle access to $g$} (out);

\draw[arrow] (Cg) to[bend right=20] node[pos=0.5, red, cross out, draw, 
    minimum width=0.4cm, minimum height=0.4cm, line width=0.6pt, 
    anchor=center, sloped] {} (out);
\draw[arrow] (Cg) to[bend right=20] node[edgelabel, below, pos=0.3, text width=2.2cm] {no examples\\from $\mathcal{C}_g$} (out);

\end{tikzpicture}
\caption{Illustration of our PAC learning setting. An auxiliary function $g$ induces a concept class $\mathcal{C}_g$. The PAC learning algorithm has oracle access to $g$ and examples from $\mathcal{C}_g$. They must output a learner/generator, which operates with oracle access to $g$, but no further examples from $\mathcal{C}_g$.}
\label{fig:pac_setting}
\end{figure}
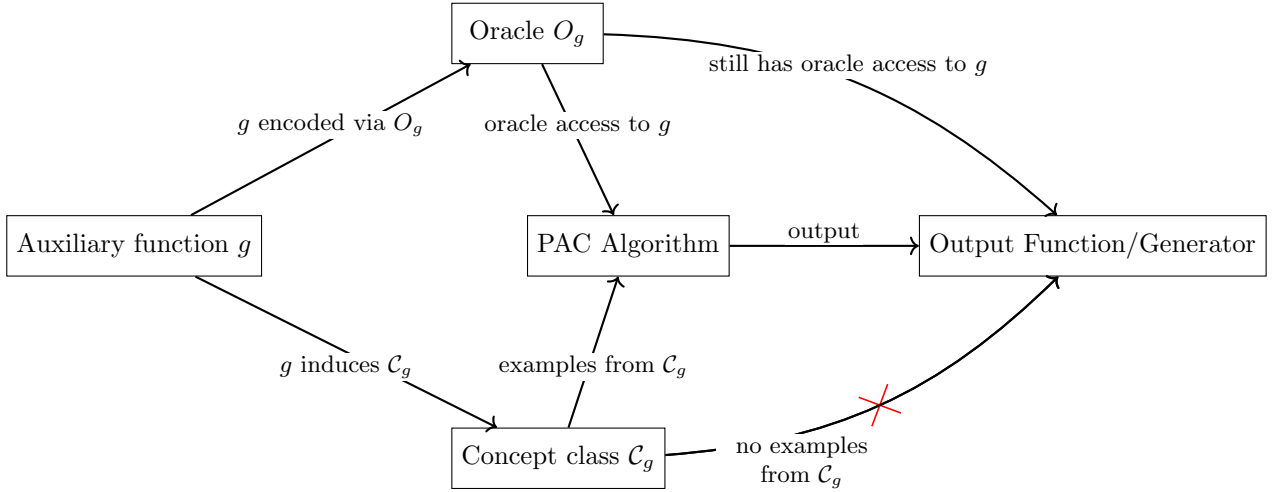
\subsection{Our Techniques}
We here provide a high-level overview of the techniques we use to obtain \cref{inf:function_hardness_does_not_imply_distribution_hardness}. To construct our concept class, we make use of \emph{one-way permutations}, which informally, are permutations that are easy to evaluate, but (computationally) hard to invert. Given a one-way permutation $g$, and a randomly sampled nonzero string $r$, one can define \emph{the hardcore predicate}, $B_r(x)=\langle g^{-1}(x),r \rangle\bmod 2$, which, roughly speaking, is a one-bit function ``encapsulating'' the hardness of $g$ (i.e., hard to compute, even given $g$). Our function class then effectively consists of all of the hardcore predicates of this one-way permutation, and the distribution class is the distribution naturally induced by these predicates. It is well-known  that evaluating the hardcore predicate given examples of the form $(g(x),r)$ is equivalent to being able to invert $g$ itself. We then equip both learners with oracle access to the one-way permutation $g$. One learner is tasked with PAC learning the concept class, whilst the other is tasked with PAC generating this concept class. We then argue that the former task cannot be done efficiently; this task indeed implies that, for \emph{most} strings $x$ sampled uniformly at random, the learner should be able to correctly evaluate $B_r(x)$: we show that any learner that can do this efficiently can be used to efficiently invert the one-way permutation $g$, contradicting well-established hardness assumptions. In contrast, we show that the latter PAC generation task, with oracle access to $g$, can be done efficiently. The crucial difference is that in the former task, for a string $x$ sampled uniformly at random, the learner effectively needs to compute $g^{-1}(x)$ for most $x$'s, which (intuitively) requires inverting $g$. However in the latter task, the algorithm only needs to generate samples from a distribution close enough to the true induced distribution, and gets to ``choose'' the $x$'s. That is, the key idea is to flip the table: the learner can uniformly sample a string $s$, and query the oracle to obtain $t \coloneqq g(s)$. Now, we have $s=g^{-1}(t)$, so the learner looks like it ``inverted'' the function on this specific point $t$, and output a sample corresponding to the string $t$, instead of $s$. Because $g$ is a one-way permutation, the distribution of $t=g(s)$ (for $s$ uniform) is still uniform. (Importantly, the generator still does not know how to invert $g$, and does not have to.) The final remark is that since  quantum-resilient one-way permutations are widely believed to exist, the hardness result extends to both quantum and classical learners.

Our second main contribution (\cref{inf:qsample_gennable_sample_not_learnable}) is a learning separation between quantum and classical examples relative to an oracle. The starting point is that, as shown by Kaplan \etal~\cite{DBLP:conf/crypto/KaplanLLN16}, Simon's algorithm still succeeds with high probability on functions with a low number of pseudoperiods. The then first sample a random function $g\colon\{0,1\}^n\to\{0,1\}^n$. From this, any nonzero string $a$ yields a periodic function $f_a$, defined by $f_a(x)=g(x)+g(x+a)$. This leads us to defining the concept class $\mathcal{C}_g$ induced by the random function $g$, which consists of all $2^n-1$ such $f_a$'s. Again, we give both learners oracle access to the random function $g$: in this case, both learners are quantum learners, the first with only classical examples of $f\in \mathcal{C}_g$, and the second with quantum examples. Both are then tasked with generating the distribution class induced by $\mathcal{C}_g$. We are able to show that, with high probability over the uniformly random function $g$, the function $f$ has a low number of pseudoperiods for all real periods $a$. By the probabilistic method, this then implies there exists a function $g$ such that all possible offsets $a$ have a low enough pseuodoperiod factor that Simon's algorithm succeeds with high probability in recovering the period $a$. This implies that the learner with quantum examples can easily generate this distribution class, by recovering the period $a$, then uniformly sampling a string $x$, and querying the oracle $g$ twice to obtain $g(x)$ and $g(x+a)$. But, by standard hardness results, recovering the period $a$ from classical examples only requires an exponential number of examples, so this particular route is unavailable to the second learner. We then show that any algorithm that does not recover the period will generate from a distribution at total variation distance $1-o(1)$ from the true distribution, implying that any quantum learner with only classical samples cannot efficiently generate this distribution class. This in turn, shows a separation between classical and quantum examples. 
\subsection{Related Work}\label{ss:related}
Previously the question of quantum advantage in learning has been rigorously studied. These have been studied in various access models, for various problems in testing and learning. They also vary from proving sample complexity separations, or time complexity separations. Most works also do not make the explicit distinction between samples and computation, as we are interested in for this work. That being said, we survey the prior work, and explicitly break down the modes of access, advantage proven, and the source of the advantage.

\paragraph{Quantum sample access. }We start with the results shown with quantum sample access. Apart from allowing the quantum learner only quantum computation, and otherwise using classical examples, this is the weakest form of quantum data access that has been studied. As seen in prior works, in most contexts, this is not sufficient to prove large quantum advantage. For PAC learning under arbitrary distributions, Arunachalam and De Wolf show that quantum learners, with quantum computation and this model of access, only admits a constant factor advantage in sample complexity over the corresponding classical learners \cite{Arunachalam2018}. Hence, for the task of general PAC learning of functions under arbitrary distributions, quantum samples are no more powerful than classical samples. One could obtain other separations, however, if one relaxed the requirement of needing to learn a concept class under an arbitrary distribution. For instance, Bshouty and Jackson gave an efficient quantum learner (in time) for learning DNF formulae over the uniform distribution using this model of access \cite{BShouty1995}. This is a quantum advantage over the best classical learner that is known, though no hardness results exist to show this is a proven quantum advantage. Turning to the easier problem of testing, rather than learning, Atici and Servedio showed that there were classes that admitted quantum sample complexity advantages. Specifically, they showed that juntas were testable with quadratic advantage \cite{DBLP:journals/qip/AticiS07}.

The above has been focused on functions, but one could also ask the same questions regarding advantages in learning and testing for distributions. Using quantum sample access, Gily\'{e}n and Li commented that advantages in sample complexity for testing distributions in this model were at most cubic, via a result of Chailloux \cite{DBLP:conf/innovations/Chailloux19}. They then gave several quadratic sample advantages using this model of access, such as for testing $\ell_2$ closeness \cite{DBLP:conf/innovations/GilyenL20}. On the other hand, the question of learning for distributions has not been studied as much with this form of access. As previously mentioned, the closest work is that of Sweke \etal, who showed there were distributions that were efficiently learnable (in time) with quantum computation but not with classical computation, though both learners in this case were limited to classical samples. Hence, our work makes progress to resolve the questions of sample advantage in this setting. Interestingly, though Gily\'{e}n and Li's results imply that for testing distributions, this type of sample access is too weak to provide strong separations in sample complexity, our results show that this weaker form of quantum samples can actually prove exponential separations in time and sample complexity, albeit relative to an oracle.

To conclude this discussion, we also mention the passive testing model, specifically for testing of functions. In this model, instead of query access to a function, one instead gets random samples from the function's truth table. Here, for tasks like testing symmetry and monotonicity, Caro \etal~ were able to prove exponential sample separations between quantum and classical learners \cite{DBLP:conf/innovations/CaroNS26}. It is worth noting the assumption of passive testing is important, since these advantages collapse to at most quadratic when queries are permitted. They also interestingly demonstrate a separation between quantum samples and classical membership queries, and exhibit a property that can be efficiently tested with classical membership queries but not with quantum samples, showing in general that the two are incomparable. Though not the focus of this paper, it would be interesting to see if that membership queries and quantum samples are still incomparable when extended to learning, either functional or distributional.  

\paragraph{Access to a generating circuit. }A stronger form of access is allowing the learner to have access to the actual circuit that generates the quantum samples. This allows them to not only produce samples themselves, but also to implement the inverse as a subroutine for their learning algorithms. Clearly, this form of access is more powerful than the prior, but the question is whether it actually leads to more powerful quantum algorithms. Immediately the answer is yes, as this form of access gives an advantage in the PAC learning setting, where Salmon \etal~showed this form of access allows a quadratic sample advantage in PAC learning under arbitrary distributions, contrary to the previous results by Arunachalam and De Wolf \cite{salmon24provable}. Gily\'{e}n and Li also showed in testing the squared Hellinger distance between distributions, this form of access gives a quantum tester a polynomial sample advantage greater than cubic \cite{DBLP:conf/innovations/GilyenL20}. This is in contrast to the previous quantum sample access model, which gave at most a cubic advantage. This form of access has also been studied in other testing models, such as in uniformity testing \cite{DBLP:conf/tqc/CanonneKO25}. While it would be interesting to quantify the advantage this specific form of access gives in distribution learning, it is also notable that we do not need this type of access to prove our exponential separation.

\paragraph{Advantages when only a description is required. }We also briefly survey work in the other output model in PAC learning, specifically, where the learners are only required to efficiently output a description of the hypothesis circuit, which need not be efficiently evaluatable. We refer to the exposition by Gyurik and Dunjko for more detail, but at a high level, they make a distinction between where the difficulty in PAC learning stems from evaluating, as opposed to the more intuitive ``learning'' task \cite{DBLP:journals/corr/abs-2306-16028}. The idea is that certain separations, such as that by Liu \etal, have difficulties which stem from the inability of the classical learner to efficiently evaluate some function \cite{liu2021rigorous}. In this sense, one can argue that it is not a ``new'' quantum advantage, or that it is also not a quantum \emph{learning} advantage, since the classical learner arguably can perform the learning portion of the task without difficulty, but then has difficulty evaluating the output. The idea then, in requiring the PAC learner to just output a description, instead of an efficient circuit, is that any quantum advantage in this setting will come from a quantum learner outperforming a classical learner in the learning task, and not the evaluation task. In this output model, Molteni \etal~introduced the task of verifiable identification, for which there exists a time efficient quantum learner, but not a time efficient classical learner, unless \textbf{BQP} is low in the polynomial hierarchy. This is specifically a quantum computation advantage, since both the quantum and classical learners are given classical data.

\paragraph{Learning and testing of quantum objects. }This work concerns itself with learning of classical objects, but there is a breadth of work on learning and testing quantum objects. Perhaps surprisingly, even though one might expect that a quantum computer and quantum data is needed to interact with these quantum objects, some of these tasks are actually efficiently doable with a purely classical learner. For example, Lewis \etal~showed an efficient classical learner exists for learning certain ground states \cite{DBLP:journals/corr/abs-2301-13169}. Thus, one could investigate the quantum objects for which there is quantum advantage to learn. Molteni \etal~showed there were certain quantum observables that exhibited a quantum computation advantage \cite{MolteniRiccardo2026Eqai}. That is, when both are given classical samples, there were observables that could be efficiently learned by a quantum learner, but not by a classical learner. The above mentioned result on verifiable identification by Molteni \etal~also have implications on learning quantum objects, such as Hamiltonian learning \cite{MolteniRiccardo2026Eqai}, which also exhibit a quantum computation advantage. To the best of our knowledge, the advantage provided by quantum computations has not been studied in this context, and could be a future direction of work.
\subsection{Outline}
We provide an overview of PAC learning and the types of oracle access we will consider in \cref{sec:preliminaries}. We then move to show that function hardness does not imply distributional hardness in PAC learning relative to an oracle in \cref{sec:function_hardness_does_not_imply_distributional_hardness}. Then, we show a separation between quantum examples and classical examples relative to an oracle in \cref{sec:quantum/classical_example_separations}. We then conclude and provide some open problems in \cref{sec:conclusion}.

\section{Preliminaries}\label{sec:preliminaries}
We use standard asymptotic notation throughout, as well as the (slightly less standard) $\tilde{O}(\cdot),\tilde{\Omega}(\cdot)$ notation to omit polylogariothmic factors in the argument. We write ``u.a.r.'' for ``uniformly at random'', and will denote addition modulo 2 by $+$ instead of $\oplus$. Finally, the concatenation of two strings $a$ and $b$ is denoted $a||b$. In what follows, we use \emph{computationally efficient} (or just \emph{efficient}) to indicate that an algorithm runs in polynomial time (in the relevant parameters).
\subsection{PAC Learning}
We introduce the necessary definitions for PAC learning in this section. All definitions will be for general oracles, with more details on the specific forms of oracles we use given in \cref{sec:oracles}. First, we recall the standard definition of PAC learning of functions: for in depth-survey of this field, we refer the reader to~\cite{DBLP:journals/corr/abs-2511-08791}.
\begin{definition}[PAC Learning of Functions]\label{def:pac_learning_functions}
Let $\mathcal{D}$ be a probability distribution over $n$-bit strings, and $\mathcal{F}_n$ be the set of all functions with domain in the $n$-bit strings. An algorithm $\mathcal{A}$ is an $(\epsilon,\delta,O,\mathcal{D})$-PAC learner for a concept class $\mathcal{C}\subseteq \mathcal{F}_n$ if for all $c\in\mathcal{C}$, when given the oracle $O(c,\mathcal{D})$, which outputs an example from $c(x)$ with $x$ drawn according to probability distribution $\mathcal{D}$, with probability at least $1-\delta$, the learner $\mathcal{A}$ outputs a hypothesis $h\in\mathcal{F}_n$ such that
\begin{equation*}
    \Pr_{x\leftarrow\mathcal{D}}[h(x)\neq c(x)]\leq\epsilon.
\end{equation*}
By outputting a hypothesis, we mean that it outputs a circuit that evaluates the hypothesis function $h$. We say such an algorithm is a \emph{PAC-learner} if it is an $(\epsilon,\delta,O,\mathcal{D})$-PAC learner for every $\epsilon,\delta$, and arbitrary (unknown) distribution $\mathcal{D}$. We note that while the general form of PAC learning requires that the learner can output such a hypothesis when given examples drawn from \emph{any} (unknown) distribution $\mathcal{D}$, sometimes, we will be interested in the case where the examples are drawn from a specific pre-defined distribution $\mathcal{D}$: an algorithm only required to succeed in this case will be said to be a \emph{PAC-learner with respect to the distribution $\mathcal{D}$}. We define the \emph{sample complexity} of the learner $\mathcal{A}$ as the number of samples taken by the algorithm, and we call the learner \emph{efficient} if it outputs a circuit in time polynomial in $n,1/\varepsilon$, and $\log(1/\delta)$, and the corresponding circuit also runs in time polynomial in $n,1/\varepsilon$ and $\log(1/\delta)$ (not necessarily the same polynomial).
\end{definition}
We will be specifically interested in the case of \emph{uniform} PAC learners, that is, PAC-learners with respect to the uniform distribution. We will also be concerned about PAC GEN learning of distributions, defined below. First, we need the concept of (approximate) generator:
\begin{definition}[Approximate generator]
Fix a distance measure $d$ between probability distributions, and $\epsilon >0$. An algorithm is said to be \emph{a $(d,\epsilon)$-generator for a probability distribution $\mathcal{D}$} if it is computationally efficient, and the distribution  $\mathcal{D}'$ of its output satisfies $d(\mathcal{D},\mathcal{D'})\leq\epsilon$.
\end{definition}
\begin{definition}[PAC GEN Learning of Distributions]\label{def:pac_gen_learning}
    Fix a distance measure $d$ between probability distributions. A learning algorithm $\mathcal{A}$ is an \emph{$(\epsilon,\delta,O,d)$-PAC GEN-learner for a distribution class $\mathcal{C}$}, if for all $\mathcal{D}\in \mathcal{C}$, when given oracle $O(\mathcal{D})$, which outputs a sample according to the distribution $\mathcal{D}$, with probability $1-\delta$, the learner $\mathcal{A}$ outputs a (classical) $(d,\epsilon)$-generator $GEN_{\mathcal{D}'}$ for $\mathcal{D}$. The generator $GEN_{\mathcal{D}'}$ has access to all the resources that the original learning algorithm $\mathcal{A}$ has, except the example oracle $O(\mathcal{D})$. Namely, if distribution class is induced by a function $g$ encoded via an oracle, the outputted generator may still use calls to the oracle $g$.
    We further say the learning algorithm \emph{efficient} if it outputs the generator in polynomial time.
\end{definition}
The notion of distance we will adopt in this paper is the total variational (TV) distance, which we expand on in \cref{sec:tv_distance}. We emphasize that the above definition provides the outputted generator $GEN_{\mathcal{D}'}$ with the same ``resources'' as the learning algorithm $\mathcal{A}$ uses, \emph{except for the example oracle itself}. This is crucial in this paper, as our hard instances rely two oracles, $O(\mathcal{D})$ (the example oracle) and oracle access to some auxiliary function: the generating algorithm will still have access to the latter, but not the former.

Thus, \cref{def:pac_learning_functions} is concerned with learning functions, while \cref{def:pac_gen_learning} is concerned with learning distributions. Accordingly, we will refer to algorithms learning functions as \emph{PAC learners}, or simply \emph{learners}, and algorithms generating distributions as \emph{PAC generators}, or simply \emph{generators}. Further, notice that the definitions do not make any distinction between quantum and classical algorithms. For the purposes of this paper, we will assume that all learners have {quantum computation}. Hence, when we refer to ``classical learners,'' we mean ``quantum algorithms with \emph{classical examples},'' and ``quantum learners'' as ``quantum algorithms with \emph{quantum examples}.''

Finally, we also note that every boolean function on $n$ variables naturally induces a distribution $D_f\in\mathcal{D}_{n+1}$, the set of distributions on $n+1$ variables, via \emph{distribution classes}.
\begin{definition}[Distribution Class]\label{def:distribution_class}
    Given a boolean function $f\in\mathcal{F}_n$, we define the distribution $D_f\in\mathcal{D}_{n+1}$ as the distribution that is uniform over $(n+1)$-bit strings of the form $x||f(x)$. Additionally, given a concept class $\mathcal{C}\subseteq \mathcal{F}_n$, we define the distribution class $\mathcal{D}_{\mathcal{C}}\subseteq\mathcal{D}_{n+1}$ as
    \begin{equation*}
        \mathcal{D}_{\mathcal{C}}=\{D_c \mid c\in\mathcal{C}\}.
    \end{equation*}
\end{definition}
\subsection{Total Variation Distance}\label{sec:tv_distance}
In this section, we formally define and give properties on TV distance. Intuitively, one can think of the TV distance as the maximal difference between how two probability distributions assign probabilities to (sets of) events. We here focus on discrete domains, for simplicity, and as this is the focus of our work.
\begin{definition}[TV Distance]
    Let $\mathcal{D}_1$ and $\mathcal{D}_2$ be two probability distributions defined over $S$. Then, the TV distance between $\mathcal{D}_1$ and $\mathcal{D}_2$ is given by
    \begin{equation*}
        \TV(\mathcal{D}_1,\mathcal{D}_2)
        = \sup_{E\subseteq S} (\mathcal{D}_1(E)-\mathcal{D}_2(E))
        =\frac{1}{2}\sum_{x\in S}|\mathcal{D}_1(x)-\mathcal{D}_2(x)|.
    \end{equation*}
\end{definition}
One can check that the TV distance is a metric, and takes values in $[0,1]$. 
We also use the following standard alternate form of the TV distance:
\begin{lemma} For any two probability distributions $\mathcal{D}_1,\mathcal{D}_2$ over $S$, we have
    \begin{equation*}
        \TV\mathcal({D}_1,\mathcal{D}_2)=1-\sum_{x\in S}\min(\mathcal{D}_1(x),\mathcal{D}_2(x)).
    \end{equation*}
\end{lemma}
\noindent See, e.g., \cite[Chapter 11]{DBLP:books/daglib/0012859} for more on total variation distance.
\subsection{Oracle Definitions}\label{sec:oracles}
In this section we provide more detail on what we mean by access to different types of examples. We will use three types of oracle access to the examples, which we define below.
\begin{definition}[Oracles]\label{def:oracles}
    The three oracles we are concerned about are the $\sample$ oracle, the quantum equivalent $\qsample$ oracle, and the $\pex$ oracle. We note for all three oracles, each call is independent across executions.
    \begin{enumerate}
        \item Given a distribution $\mathcal{D}$, the $\sample$ oracle returns some $x$ drawn from $\mathcal{D}$. The outputs across different calls of the oracle are independent.
        \item Given a distribution $\mathcal{D}$, the $\qsample$ oracle returns a superposition of all elements $x$ in the distribution, $\sum_x\sqrt{\mathcal{D}(x)}|x\rangle$.
        \item Given a distribution $\mathcal{D}$ and a function $f$, the random example oracle $\pex$ first samples $x$ according to the distribution $\mathcal{D}$, then returns the tuple $(x,f(x))$. The outputs across different calls of the oracle are independent.
    \end{enumerate}
\end{definition}
That is, the $\sample$ oracle returns a random sample drawn from the distribution $\mathcal{D}$, and the $\qsample$ oracle returns a ``quantum example'' from the same distribution, which is a superposition according to the weights of the distribution.\footnote{In particular, measuring this state in the computational basis yields a classical sample from the distribution $\mathcal{D}$.} On the other hand, the $\pex$ oracle is an example oracle, which samples an input $x$ according to the distribution $\mathcal{D}$, and returns the tuple consisting of the sample $x$ and the evaluation $f(x)$.

\subsection{One-way functions}
Our first result will rely on the use of \emph{one-way permutations}: in this subsection, we recall some cryptographic definitions and relevant background, and refer the interested reader to \cite[Chapter 1]{DBLP:books/cu/Goldreich2001} for a thorough coverage. We start with the definition of a \emph{one-way function}.
\begin{definition}[One-way functions]\label{def:one_way_function}
    A function $f\colon\{0,1\}^n\rightarrow\{0,1\}^n$ is a \emph{one-way function} if:
    \begin{enumerate}
        \item It is \emph{easy to compute}. That is, there exists a deterministic polynomial-time algorithm that on input $x$, outputs $f(x)$.
        \item It is \emph{hard to invert}. That is, for every probabilistic polynomial-time (PPT) algorithm $\mathcal{A}$, every polynomial $p(n)$ and all sufficiently large $n$,
        \begin{equation*}
            \Pr_{x\leftarrow\{0,1\}^n}[\mathcal{A}(1^n,f(x))\in f^{-1}(f(x))]<\frac{1}{p(n)}.
        \end{equation*}
        That is, for most $y=f(x)$, no PPT algorithm can find any preimage of $y$ with non-negligible probability.
    \end{enumerate}
\end{definition}
A \emph{one-way permutation} is then an injective one-way function. Given a one-way function, we can also define the related \emph{hardcore predicate}.
\begin{definition}[Hardcore Predicate]\label{def:hardcore_predicate}
    A polynomial computable predicate $b:\{0,1\}^*\rightarrow\{0,1\}$ is a \emph{hardcore} of a function $f$ if for every PPT algorithm $A'$, every polynomial $p(\cdot)$, and all sufficiently large $n$,
    \begin{equation*}
        \Pr_x[A'(f(x))=b(x)]<\frac{1}{2}+\frac{1}{p(n)}.
    \end{equation*}
\end{definition}
We know of the existence of a specific form of hardcore predicate which we will use, for example, by \cite[Theorem 2.5.2]{DBLP:books/cu/Goldreich2001}. Any one way function can be converted into the \emph{hardcore} form.
\begin{theorem}\label{thm:goldreich_hardcore_predicate}
    Let $f$ be a one-way function,and define $g(x,r)$ as the tuple $g(x,r)=(f(x),r)$, where $r$ is a string such that $|x|=|r|$. Let $b(x,r)$ denote the inner product of the binary vectors $x$ and $r$. Then the predicate $b$ is a hardcore of the function $g$.
\end{theorem}
It is worth noting the string $r$ is not kept hidden from the adversary. We then have the following standard result, which shows that predicting a predicate implies inverting the one-way function: i.e., the predicate captures the computational hardness. 
\begin{theorem}[Goldreich-Levin, \cite{DBLP:conf/stoc/GoldreichL89}]\label{thm:hardness_of_inverting}
    If there exists a PPT algorithm that predicts $b(x,r)$ when given a tuple $(f(x),r)$ with with probability $\frac{1}{2}+\eta$, then there exists an algorithm that inverts $f$ with probability at least $\poly(\eta)/\poly(n)$, and runs in time $\poly(n,1/\eta)$. We refer to $\eta$ as the \emph{advantage}, and refer to it as non-negligible if it is at least $1/\poly(n)$ in magnitude.
\end{theorem}
Hence, if $\eta$ non-negligible, there exists a PPT algorithm which can invert the hardcore predicate. We conclude the section by noting that it is not known whether one-way functions exist; in fact, showing their existence would imply $\text{\textbf{P}}\neq\text{\textbf{NP}}$. Furthermore, the existence of one-way functions alone also does not imply the existence of one-way permutations, and the latter is a stronger assumption. However, their existence is widely conjectured, and constitutes a standard cryptographic assumption.

\section{Function Hardness Does Not Imply Distributional Hardness}\label{sec:function_hardness_does_not_imply_distributional_hardness}
\subsection{Context Of Prior Work by Sweke \etal}\label{sec:sweke_context}
In this section, we summarize a potential strategy given by Sweke \etal~that might be used to prove a separation between quantum and classical PAC learning. One key result they were able to prove was that if a concept class $\mathcal{C}$ is efficiently PAC-learnable with respect to the uniform distribution, then the induced distribution class $\mathcal{D}_C$ is also efficiently PAC-GEN learnable.
\begin{theorem}[Function Learnability Implies Distribution Generation, \cite{DBLP:journals/quantum/SwekeSHE21}]\label{thm:function_learnability_implies_distribution_learnability}
    If a concept class $\mathcal{C}$ is efficiently classically (resp. quantumly) PAC learnable with respect to the uniform distribution and the $\pex$ oracle, then the distribution class $\mathcal{D}_\mathcal{C}$ is efficiently classically (resp. quantumly) PAC GEN-learnable with respect to the $\sample$ oracle and TV distance.
\end{theorem}
They left the reverse direction open as a conjecture.
\begin{conjecture}[Function Hardness Implies Distribution Hardness, \cite{DBLP:journals/quantum/SwekeSHE21}]\label{conj:function_hardness_implies_distribution_hardness}
    If a concept class $\mathcal{C}$ is not efficiently classically (quantumly) PAC learnable with respect to the uniform distribution and the $\pex$ oracle, then the distribution class $\mathcal{D}_\mathcal{C}$ is not efficiently classically (quantumly) PAC GEN-learnable with respect to the $\sample$ oracle and TV distance.
\end{conjecture}
The hope was that if this conjecture was true, and one could find a boolean concept class $\mathcal{C}$ such that:
\begin{enumerate}
    \item $\mathcal{C}$ is not efficiently classically PAC learnable, with respect to the uniform distribution and $\pex$ oracle, and
    \item $\mathcal{C}$ is efficiently quantumly learnable, with respect to the uniform distribution, and the $\pex$ oracle,
\end{enumerate}
then this would immediately imply that the induced distribution class $\mathcal{D}_\mathcal{C}$ would not be efficiently classically PAC GEN-learnable, yet would be efficiently quantumly PAC GEN-learnable. Though it is currently not known whether a concept class does satisfy the above two criteria, proving the reverse direction of this conjecture true would allow one to instantly leverage such a concept class, if it exists, into a PAC GEN-learning separation. Then, should this hold true, Sweke \etal~conjectured that one could potentially modify both of the above results to prove hardness separations for classical and quantum \emph{samples}, instead of classical and quantum computation separations, for boolean concept classes and their induced distributions, though this was nontrivial, and would take a lot more work to make sound.

In the remainder of this section, we establish~\cref{inf:function_hardness_does_not_imply_distribution_hardness} (formally restated as~\cref{thm:function_hardness_does_not_imply_distribution_hardness}), providing strong evidence that this reverse direction does not hold: i.e., we construct an explicit counterexample showing this conjecture to be false relative to an oracle.

\subsection{A Counterexample to \cref{conj:function_hardness_implies_distribution_hardness}}
Assuming the existence of one-way permutations, we show~\cref{conj:function_hardness_implies_distribution_hardness} to be false by explicitly constructing a function class which is not efficiently learnable, but induces a distribution that is efficiently PAC-gen learnable (and in fact, exactly so). First, note given any function $f(x)$, one can naturally define a distribution class given by the tuples $(x,f(x))$, via \cref{def:distribution_class}. We thus create a distribution class consisting of a family of functions $\{f(x)\}$, which then induces a family of distributions over tuples $\{(x,f(x)\}$, and show whilst the former cannot be efficiently learned, the latter can be efficiently generated.
\begin{theorem}\label{thm:function_hardness_does_not_imply_distribution_hardness}
    Assume one-way permutations exist, and let $\mathcal{W}$ be a (classical) oracle providing query access to a one-way permutation. Then there is a concept class $\mathcal{C}$ that can be efficiently generated with respect to the $\sample$ oracle and TV-distance, but is not efficiently learnable with respect to the uniform distribution and the $\pex$ oracle.
\end{theorem}
We make a slight alteration to the PAC learning requirement in \cref{def:pac_learning_functions}, strengthening the outputted hypothesis. Recall in the classical definition, the learner must output a hypothesis $h(x)$ that only errors with probability $\epsilon$. However, we will allow the learner to output a function $h(\mathcal{W},x)$ that can be efficiently evaluated. In other words, the learner can output a circuit, which in addition to the input, can make a polynomial number of calls to the permutation oracle, with the aim that this circuit only errors with probability $\epsilon$. The reason for doing so is that our PAC GEN learner relies on access to this permutation oracle to be efficient. If we did not give the hypothesis function the same type of access, one might argue that the separation is because the generator has access to additional resources, being oracle queries in its generating process, whilst the hypothesis output does not. In this case, we allow the learner to output something comparable, and prove that even if the learner is able to output a hypothesis which can incorporate a polynomial amount of oracle calls, \cref{thm:function_hardness_does_not_imply_distribution_hardness} still holds. We now move to the proof.
\begin{proof}
    Notice all queries are classical under the given oracle access. We will show the result for classical computers, and the argument for quantum computers is identical (as discussed in the remarks after the proof). We let $g:\{0,1\}^{n-1}\rightarrow \{0,1\}^{n-1}$ be a one-way permutation to which the algorithms have oracle access. In other words, there is an oracle which the algorithm can query at any time with a string $x$, and the oracle will return the mapping of $x$ under the one-way permutation $g(x)$. Next, let $r\in\{0,1\}^{n-1}$ be a uniformly sampled nonzero vector, and let $B_r(x):\{0,1\}^{n-1}\rightarrow\{0,1\}$ be the function $B_r(x)=\langle x,r\rangle\bmod 2$, the hardcore predicate of the function $l(x)=(g(x),r)$. Note that, by \cref{thm:hardness_of_inverting}, being given $l$ and inverting $B_r$ is equivalent to inverting $g$. Given a specified one-way permutation $g$, and a randomly sampled $r$, we define our boolean function $f:\{0,1\}^n\rightarrow\{0,1\}$ as follows. For ease of representation, instead of taking the input as an $n$-bit string, we will break the input into two strings $t\in\{0,1\}^{n-1}$ and $b\in\{0,1\}$, so that $x=t||b$. Then $f$ is given by:
    \begin{equation*}
        f(x)=f(t||b)=\begin{cases}
            B_r(t)=\langle t, r\rangle\mod 2,&\text{if }b=0,\\
            B_r(g^{-1}(t))=\langle g^{-1}(t),r\rangle\mod 2,&\text{if }b=1.
        \end{cases}
    \end{equation*}
    Then, given a one-way permutation $g$, our concept class $\mathcal{C}$ is all such $f$ obtained from the one-way permutation $g$ and the $2^{n-1}$ possible random strings $r$. The corresponding distribution generated is the uniform one over strings $x||f(x)$.

    First, notice that in both the learning and generating setting, the algorithm can efficiently learn the hidden string $r$: indeed, to do so, it suffices to find $n-1$ linearly independent nonzero strings $t$ with $b=0$ in their random samples, and then solve a system of equations; with overwhelming probability, the former will occur in a polynomial amount of samples. However, we argue that even though learning $r$ is easy, it is hard to learn the function $f$. To see this, assume an algorithm was able to predict $f(x)$ with non-negligible advantage. Then, they must be able to predict $f(x)$ for strings where $b=1$ with non-negligible advantage. Under the $\pex$ oracle given to the learner, the algorithm is provided with an example $g(t)$ by taking the first $n-1$ bits of $x$. Combined with having already learned $r$, this gives them a tuple $(g(t),r)$. Predicting $f(x)$ with advantage in this case means that they can predict $B_r(g^{-1}(t))$ with advantage. To see why, assume an algorithm was able to predict $f$ with non-negligible advantage over a uniformly random choice of $x=t||b$. Then, as we will make rigorous shortly, they must in particular be able to predict $f(x)$ for strings of the form $x=t||1$ (i.e., with $b=1$), which, even knowing $r$, means predicting with non-negligible advantage the hardcore predicate $\langle s, r\rangle$ of $g$ as defined in \cref{thm:goldreich_hardcore_predicate}, given a tuple $(g(s),r)$ (where $s \coloneqq g^{-1}(t)$).~--~contradicting that $g$ is a one-way permutation. (Note that even a polynomial number of samples from the permutation oracle $\mathcal{W}$ would still not help, unless by chance the preimage of $t$ was in these samples: which only happens with probability $o(1)$, where the $o(1)$ term decays exponentially in $n$.) We now  proceed to formalize this intuition.
    
    The fact that $g$ is a one-way permutation implies, for any efficient hypothesis $h(\mathcal{W},x)$ that the learner produces, the probability that it agrees with the hardcore predicate is bounded by $\Pr_x[h(\mathcal{W},x)=b(x)]<\frac{1}{2}+\frac{1}{p(n)}$ for every $p(n)$. The hypothesis is required to predict the hardcore predicate half of the time, where the bit string ends in 1. Hence, we get that
    \begin{equation*}
        \Pr[h(\mathcal{W},x)=f(x)]\leq \frac{3}{4}+\frac{1}{p(n)},
    \end{equation*}
    where the probability is taken over the uniform distribution. As such, this cannot be efficiently PAC learned for $\epsilon<0.25$.

    However, notice the generator has an effective way of generating this distribution, and in fact, it can generate the \emph{exact} distribution. First, it learns $r$. Second, it samples the random strings $t$ and $b$. If $b=0$, it computes $B_r(t)$, having learned the string $r$. Then, it can output $t||0||B_r(t)$ as a sample. Otherwise, if $b=1$, it can feed the string $t$ into the one-way permutation oracle, to obtain the string $s\coloneqq g(t)$: it can now compute $B_r(t)=B_r(g^{-1}(s))$, and output the string $s||1||B_r(t)$. Since $g$ is a one-way permutation, uniformly sampling $t$ is equivalent to uniformly sampling $s$. The learning algorithm just needs to learn $r$ (which, as outlined above, can be done with high probability using a polynomial amount of queries) before outputting a generator $GEN$ which samples from the distribution in constant time (and at most one call to the permutation oracle) as described above. This proves the theorem for classical computers.
\end{proof}
We make three concluding remarks here. First, note that even when given quantum computation, one-way functions are still assumed to be hard to invert, assuming one uses a quantum secure one way permutation. As such, the above results also generalize to hold even for quantum algorithms. Furthermore, while our construction only assumes classical access to the one-way permutation, the results also hold when the generating algorithms are given \emph{both} quantum and classical access to the one-way permutation (since one-way permutations are still assumed to be hard to invert given this type of access). By quantum access, we mean being given access to a uniform superposition $\sum_x|x\rangle|g(x)\rangle$ over the output of $g$. Note though, as part of our construction, the generating algorithm needs to explicitly evaluate a uniformly sampled input in the one-way permutation, and thus requires at least classical access: having \emph{only} quantum access to the one-way permutation does not suffice for our argument. Finally, we explicitly need one-way permutations, and not just one-way functions, since we need the preimage of a uniformly random element to also be uniformly distributed. We leave it as an open question to relax this requirement to only needing one-way functions, or indeed a separate construction that does not require this at all.

\section{A Quantum/Classical Separation in PAC GEN Learning}\label{sec:quantum/classical_example_separations}
After having provided strong evidence that \cref{conj:function_hardness_implies_distribution_hardness} is false in the previous section, we now establish a separation between quantum and classical examples. At a high-level, our approach for this second separation is as follows. First, we sample a random boolean function $g$, afterwards available to the algorithm as an oracle. Then, we uniformly at random sample a nonzero vector $a$, and create a function $f$ with period $a$ by defining $f(x)=g(x)+g(x+a)$, where $+$ is bitwise XOR (bitwise mod 2 addition). Now, while $f$ clearly has $a$ as period, it might also have pseudoperiods,\footnote{A pseudoperiod for an input $x$ (and an $a$-periodic function $f$) is any string $t\neq a$ such that $f(x)=f(x+t)$.} and is not guaranteed to be perfectly 2-to-1. Nevertheless, we show that the number of pseudoperiods in this function is low (with high probability), which as a result means we can still use Simon's algorithm to recover the period with high probability \cite{DBLP:conf/crypto/KaplanLLN16}. We proceed with formal definitions and proofs now.
\subsection{Probably Approximate Simon's (PAS) Functions}
We begin by defining the main building block of our concept class.
\begin{definition}[Probably Approximate Simon's Function]\label{def:PAS}
    Let $g:\{0,1\}^n\rightarrow\{0,1\}^n$ be a u.a.r function, and $a$ be a uniformly sampled random vector. We can construct a function $f$ by imposing $g$ onto itself via a shift of $a$, $f(x)=g(x)+g(x+a)$. We call such a function $f$ a \emph{probably approximate Simon's function} (PAS function).
\end{definition}
As previously discussed,t this function is guaranteed to have period $a$ (i.e., $f(x)=f(x+a)$ for every $x$), but may also contain more collisions of this form for specific $x$. We can obtain a measure of these unwanted collisions via the \emph{pseudoperiod factor}.
\begin{definition}[Pseudoperiod Factor]\label{def:pseudoperiod_factor}
    Given a PAS function with period $a$, the \emph{pseudoperiod factor} $\varepsilon(f,a)$ is the quantity
    \begin{equation}
        \varepsilon(f,a)=\max_{t\in\{0,1\}^n\setminus\{0,a\}}\Pr_x[f(x)=f(x+t)].
    \end{equation}
    We also say that input $x$ admits a pseudoperiod $t$ if $f(x)=f(x+t)$.
\end{definition}
This parameter was introduced by Kaplan \etal~\cite{DBLP:conf/crypto/KaplanLLN16} to quantify how far a function is from being perfectly 2-to-1. We will prove the following.
\begin{lemma}\label{lemma:pseudoperiod_bound}
    Let $g\colon\{0,1\}^n\to\{0,1\}^n$ be a u.a.r.\ function. For each nonzero string $a$, define $f_a$ by $f_a(x)=g(x)+g(x+a)$ for all $x\in \{0,1\}^n$. Then, with probability $1-o(1)$ over the choice of $g$, for all nonzero $a$, we have that $\varepsilon(f_a,a)=O(\frac{n}{2^n\log n})$.
\end{lemma}
We will prove this lemma at the end of the subsection; before this, we introduce the definitions and results necessary to the proof.
First, observe that if a given input $x$ admits a pseudoperiod $t$, we can rewrite this in terms of the original random function $g$: since $f(x)=f(x+t)$, we have that $g(x)+g(x+a)=g(x+t)+g(x+t+a)$, or equivalently, that $g(x)+g(x+a)+g(x+t)+g(x+a+t)=0$. Setting $\Delta_x\coloneqq g(x)+g(x+t)+g(x+a)+g(x+a+t)$, admitting a pseudoperiod is equivalent to having $\Delta_x=0$. Further, notice that this now defines a quadrilateral across the boolean cube with edges $a$ and $t$, illustrated in \cref{fig:period_quadrilateral}. Furthermore, any two quadrilaterals with differing $x$ values are either completely disjoint or coincide completely on the same face, meaning this partitions the boolean cube into $2^{n-2}$ distinct faces.
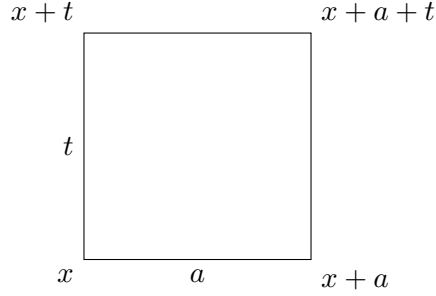
\begin{figure}[h]
    \centering
        \begin{tikzpicture}
            \coordinate (A) at (0,0);
            \coordinate (B) at (3,0);
            \coordinate (C) at (3,3);
            \coordinate (D) at (0,3);
            
            \draw (A) -- (B) -- (C) -- (D) -- cycle;
            
            \node[below left] at (A) {$x$};
            \node[below right] at (B) {$x+a$};
            \node[above right] at (C) {$x+a+t$};
            \node[above left] at (D) {$x+t$};
            
            \node[below] at (1.5,0) {$a$};
            \node[left] at (0,1.5) {$t$};
        \end{tikzpicture}
        \caption{The quadrilateral spanned by the pseudoperiod $t$ and the period $a$.}
    \label{fig:period_quadrilateral}
\end{figure}

We now follow the analysis of Daemen and Rijmen \cite{DBLP:journals/jmc/DaemenR07}, who derived concentration bounds on $\varepsilon$ for singular shifts (i.e., the probability that for a random function, $g(x)=g(x+t)$). We proceed to adapt this analysis for double shifts.
\begin{definition}[Cardinality of the differential]\label{def:cardinality_of_the_differential}
    For a fixed $g$, $a$ and $t$, the \emph{cardinality of the differential} $N(g,a,t)$ is the number of faces with edges $a$ and $t$ such that $g$ sums to zero across the vertices:
    \begin{equation*}
        N(g,a,t)=|\{\{x,x+a,x+t,x+a+t\}\mid g(x)+g(x+a)+g(x+t)+g(x+a+t)=0\}|.
    \end{equation*}
\end{definition}
Note that $N(g,a,t)=\frac{|\Delta_x|}{4}$, where $|\Delta_x|$ is the count of the number of inputs $x$ such that $\Delta_x=0$. This is because for each face that sums to zero, there are four inputs $x$ for which $\Delta_x=0$. We can show the following about the distribution of $N(g,a,t)$.
\begin{lemma}\label{lemma:cardinality_is_stochastic}
    For fixed $a,t$ and $g:\{0,1\}^n\rightarrow\{0,1\}^n$ chosen u.a.r., $N(g,a,t)$ is a random variable with binomial distribution:
    \begin{equation*}
        \Pr[N(g,a,t)=i]=(2^{-n})^i(1-2^{-n})^{2^{n-2}-i}{2^{n-2}\choose i}, \qquad 0\leq i \leq 2^{n-2}\,.
    \end{equation*}
    That is, $N(g,a,t)\sim\Bin(2^{n-2},2^{-n})$.
\end{lemma}
\begin{proof}
    A random function maps the $2^n$ different input values $x$ to independent output values $g(x)$, and hence for each quadrilateral $\{x,x+a,x+a+t,x+t\}$, it maps the output over the vertices of the quadrilateral to independent values too, and thus the sum of outputs to independent values. Each of these mappings is successful if it maps to exactly $0$, which happens with probability $\frac{1}{2^n}$, and we have $2^{n-2}$ mappings, corresponding to the number of distinct quadrilaterals.
\end{proof}
We can now complete the proof of \cref{lemma:pseudoperiod_bound}.
\begin{proof}[Proof of~\cref{lemma:pseudoperiod_bound}]
    For a fixed $t$, by \cref{lemma:cardinality_is_stochastic}, we know that $N(g,a,t)$ has distribution $Y\sim\Bin(2^{n-2},2^{-n})$, and thus has expectation $\mu=1/4$. By standard Chernoff bounds and setting $(1+\delta)\mu=4m$
    \begin{align*}
        \Pr[Y\geq (1+\delta)\mu]&\leq \left(\frac{e^\delta}{(1+\delta)^{1+\delta)}}\right)^\mu\\
        \Pr[Y\geq 4m]&\leq \left(\frac{e^{4m-1}}{(4m)^{4m}}\right)^{1/4}\\
        &= e^{m-1/4}(4m)^{-4m}\\
        &\leq \left(\frac{e}{4m}\right)^m.
    \end{align*}
    Let $T=2^n$. We note there are at most $T$ possible periods $a$, and for each of these periods $a$, at most $T$ possible pseudoperiods. This means there are at most $T^2$ pairs of period and possible pseudoperiod. Let $m$ be the number of ``bad'' trials, in the sense that $a$ is a period and $t$ is a pseudoperiod (as opposed to only a possible one). By standard Chernoff bounds, we have that
    \begin{equation*}
        \Pr[\max_t Y_t\geq m ]\leq \Pr\left[\bigcup_{t=1}^T\{Y_t\geq m\}\right]\leq T\Pr[Y\geq m]\leq T\left(\frac{e}{4m}\right)^m.
    \end{equation*}
    We can now choose the rightmost value to be small. Choose $m=\frac{8\log T}{\log \log T}$. Then, we have
    \begin{align*}
        \log[T(e/4m)^m]&=2\log T-m\log(4m/e)\\
        &=\log T-m(\log\log T-\log\log\log T+\log(32/e))\\
        &\leq2\log T-m(0.5\log\log T)\\
        &\leq -2\log T.
    \end{align*}
    This implies that $\Pr[\max_tY_t\geq m]\leq T^{-2}$, and hence, we have that $\max_t N(g,a,t)\leq O(n/\log n)$ with probability $1-T^{-2}$. Then, for each $a$, dividing by the number of possible $x$ values, we get that $\varepsilon(f,a)=O\left(\frac{n}{2^n\log n}\right)$, with probability ${1-2^{-2n}}=1-o(1)$.
\end{proof}
\subsection{Hardness Construction}
We will make use of the analysis of the Approximate Simon's Promise problem by Kaplan \etal~ to show that Simon's algorithm for sufficiently small pseudoperiods will still return a correct period with high probability.
\begin{theorem}[Theorem 1, \cite{DBLP:conf/crypto/KaplanLLN16}]\label{thm:approximate_promise_simons}
    If $\varepsilon(f,a)\leq p_0<1$, then Simon's algorithm returns $a$ with $cn$ queries with probability at least $1-(2(\frac{1+p_0}{2})^c)^n$, where $c>0$ is a parameter capturing how many times the algorithm is run to boost probability.
\end{theorem}
 Let $g\colon\{0,1\}^n\to\{0,1\}^n$ be a function encoded via an oracle. This function serves as the base of our concept class $\mathcal{C}$, which we now describe how to construct.
\begin{enumerate}
     \item For each nonzero string $a$, define the corresponding PAS function constructed from $g$. There will be $2^n-1$ of these.
     \item Each of these PAS functions $f_a$ induces a distribution over strings $x||f_a(x)$, with $x$ uniform over $n$-bit strings.
     \item These $2^n-1$ distributions form the concept class $\mathcal{C}$.
\end{enumerate}
One of these distributions will be picked as the one that needs to be PAC-generated, with one learner given classical $\sample$ access, and one given quantum $\qsample$ access. We then show the following.
\begin{theorem}\label{thm:qsample_gennable_sample_not_learnable}
    For the concept class $\mathcal{C}$ described above, with additional access to the function $g$ encoded via oracle $\mathcal{W}$, there is a quantum algorithm with $\qsample$ access to $f$ that can efficiently PAC-generate this concept class. However, no quantum algorithm with only $\sample$ access to $f$ can efficiently PAC-generate this concept class.
\end{theorem}
\begin{proof}
    We first outline how this can efficiently be done with quantum examples. First, we choose a u.a.r. function $g$ such that 
    \[
        \varepsilon(f_a,a) = O\!\left(\frac{n}{2^n\log n}\right)
    \]
    holds simultaneously for all non-zero $a$, which is guaranteed to exist by~\cref{lemma:pseudoperiod_bound} (and the probabilistic method). We use this function $g$ to define our concept class $\mathcal{C}$ as described above. Then, choosing one such $f$ in this concept class, given uniform superposition examples from $\qsample$, and for some appropriately large constant $c$, one can run Simon's algorithm via \cref{thm:approximate_promise_simons} to obtain the period $a$. This succeeds with probability $(1-(2(\frac{1+p_0}{2})^c)^n)$ and requires $cn$ samples. Then, conditional on obtaining the period $a$, the distribution can now be generated from exactly: this is done by uniformly sampling $x$, then querying $g$ on inputs $x$ and $x+a$, to calculate $f(x)=g(x)+g(x+a)$ before outputting $x||f(x)$. (Observe that although the generating algorithm still requires oracle access to the function $g$, crucially it no longer needs access to the example oracle for $f$, meaning it can successfully PAC generate this distribution class.) This algorithm only errs if Simon's algorithm fails to find the period $a$: the probability of failure can be made arbitrarily small by increasing the constant $c$.
 
    We now argue that with only \emph{classical} $\sample$ examples, even a quantum algorithms cannot generate this distribution efficiently. First, we note that obtaining the period $a$ with classical samples requires exponential sample complexity: this follows from the classical \emph{query} bound of $\Omega(\sqrt{2^n})$ due to Cleve~\cite{cleve2011classical}, meaning any classical algorithm requires at least that many queries to uncover the period $a$.\footnote{We note that there is a slight technicality, with this being a query bound, whereas in this setting we are using classical examples. However, the former is strictly stronger than the latter, since in the query setting, an algorithm can choose which $x$ to query, whereas in the generating setting the $\sample$ oracle randomly gives examples, taking choice away from the algorithm.} Then, any quantum learner, if it only has classical examples, can be simulated (possibly very inefficiently) by a classical computer with same sample complexity: this possibly computationally inefficient simulation yields a classical algorithm to which the query lower bound applies. 
    
    Having shown that finding $a$ cannot be done efficiently, it is sufficient to argue that, without learning $a$, no generator can sample from a distribution close to the required one. To do so, the first thing we argue is that a classical learner does not gain too much information from the guarantee that the random function $g$ will create a concept class where every $f$ has a small pseudoperiod factor. This is done by observing that via \cref{lemma:pseudoperiod_bound}, any random function $g$ satisfies this with probability $1-o(1)$, so even if the classical algorithm were to condition on such random functions, they would not meaningfully reduce the search space.
    
     Assume the generator produces samples from a distribution $\mathcal{D}$, and the true distribution is $\mathcal{D}_f$. We want to bound the TV distance between the two distributions. Assume that the learning algorithm sees $q=\poly(n)$ examples, and (as not charging duplicate samples to the algorithm only strengthens the lower bound) assume they are distinct. We can assume without loss of generality that the distribution that the algorithm generates from $\mathcal{D}$ is uniform over inputs $x$. In other words, the algorithm uniformly at random samples $x$, and tries to guess $f(x)$. (This is without loss of generality because the real distribution $\mathcal{D}_f$ is uniform over $x$, so having $\mathcal{D}$ not uniform only increases TV distance.) For the sake of argument, we will assume that if $x+a$ is in the queries, the algorithm will always succeed, by being able to query $g$ on both points, and thus reconstruct $f$ (this is only beneficial to the algorithm). Else, with no knowledge of $x+a$, $f(x)$ behaves like a uniformly random $n$-bit string, and the learner can only guess it with probability $\frac{1}{2^n}$. Hence, the number of points on which $\mathcal{D}_f$ and $\mathcal{D}$ agree is $\poly(n)$ with high probability, and hence the probability mass on which they agree is with high probability $\frac{\poly(n)}{2^n}$. Then, the TV distance is bounded by:
    \begin{align*}
        \TV(\mathcal{D},\mathcal{D}_f)&=1-\sum_x(\mathcal{D}(x)\wedge\mathcal{D}_f(x))\\
        &\geq 1-\frac{\poly(n)}{2^n}\\
        &=1-2^{-\Omega(n)}.
    \end{align*}
    Overall, with high probability, if the learner does not know $a$, it will generate samples from a distribution that has TV distance close to 1, showing that no quantum algorithm with only $\sample$ access can efficiently PAC-generate this concept class.
\end{proof}
\begin{remark}
    We note that embedding the period in such a way that can be recovered via Simon's algorithm is crucial. This is because in order to require the period, Simon's algorithm only requires uniform superpositions of $f$. A stronger result might be to try using Boneh and Lipton's period embedding, as described by Zhandry in \cite{10.1145/3450745}. Indeed, this would create such a separation using all quantum secure functions. However, crucially Boneh and Lipton's period embedding algorithm requires the ability to do more than just uniform superpositions. Among other things, the learner must be able to evaluate $f$ at specific points via classical queries (see, Section 6 in \cite{DBLP:conf/crypto/BonehL95}), which cannot be done in this model, since $\qsample$ does not actually give query access, and only gives the uniform superposition example. An attempt to generalize this separation must only rely on querying $f$, or the induced distribution in uniform superposition.
\end{remark}
\section{Conclusion}\label{sec:conclusion}
We gave an explicit distribution class that, relative to an oracle, can be efficiently generated by a quantum algorithm with quantum examples, but not by a quantum algorithm with only classical examples. To the best of our knowledge, this is the first separation between quantum and classical examples shown in the learning setting. We conclude with a few avenues for future research.
\paragraph{Oracles.} In both our constructions for \cref{thm:function_hardness_does_not_imply_distribution_hardness} and \cref{thm:qsample_gennable_sample_not_learnable}, we define the PAC learning instance with regards to an auxiliary function, which induces the concept and distribution classes. We then allow both the learner and the outputted generator to query this auxiliary function by providing oracle access to it. It would be interesting to see if a hard instance could be constructed in either of these instances without having the concept or distribution classes being induced by this auxiliary function. Or as an intermediate, it would be interesting to see if hard instances could be constructed where though the concept and distribution classes are induced by this auxiliary function, the outputted generators do not require oracle access to this function. On the other hand, we think that studying what oracle separations can be proven for learning and testing is also an interesting and underexplored area that may be worth pursuing.
\paragraph{Relaxed Assumptions for \cref{thm:function_hardness_does_not_imply_distribution_hardness}.} In our proof of \cref{thm:function_hardness_does_not_imply_distribution_hardness}, we assume the existence of one-way permutations. Though this is a standard assumption in cryptography, we wonder if such a strong assumptions is necessary. In other words, it would be interesting to see if a hard instance that relies on weaker assumptions could be constructed. For example, is there a hard instance that only requires assuming the existence of one-way functions, instead of one-way permutations? Or is there a hard instance that does not rely on one-way functions, but possibly another cryptographic primitive, such as trapdoor functions.
\paragraph{Maximal Quantum and Classical Example Separations.} Implicitly, in the proof of \cref{thm:qsample_gennable_sample_not_learnable}, we show that there exists a distribution class that can be generated in $O(n)$ quantum examples, but cannot be generated with less than $\Omega(2^{n/2})$ classical examples. This result just follows directly from the sample bounds in Simon's problem. It would be interesting to see what the maximal separation between quantum examples and classical examples in this setting could be. For example, Aaronson and Ambainis \cite{DBLP:journals/siamcomp/AaronsonA18} show that in regards to the sample complexity for boolean functions, the forrelation problem can be solved in $1$ quantum query, but requires $\widetilde\Omega(2^{n/2})$ classical queries. It would be interesting to see if such a separation between quantum and classical examples in the learning setting could be shown as well.

\paragraph{Acknowledgements.} We would like to thank Cl\'ement Canonne for comments, feedback and proofreading for this writeup and for discussions on the work. We would like to thank Ryan Sweke for discussions on their previous work, comments on this work, particularly on the introduction and related works section, and for pointing to \cite{DBLP:conf/innovations/GilyenL20}. We would like to thank Matthias Caro for comments on this work, and for pointing us to \cite{ICLR2026_f0075fe4} and \cite{MolteniRiccardo2026Eqai}. AI assistance was used in generating \cref{fig:pac_setting}.
\printbibliography
\end{document}